\documentclass[sigconf]{acmart}
\usepackage[noorphans]{quoting}
\usepackage{subfigure}
\newcommand{\tabincell}[2]{\begin{tabular}{@{}#1@{}}#2\end{tabular}}
\newcommand\at[2]{\left.#1\right|_{#2}}

\AtBeginDocument{%
  \providecommand\BibTeX{{%
    \normalfont B\kern-0.5em{\scshape i\kern-0.25em b}\kern-0.8em\TeX}}}

\setcopyright{acmcopyright}
\copyrightyear{2021}
\acmYear{2021}
\acmDOI{10.1145/3459637.3482264}

\acmConference[CIKM '21]{CIKM '21: ACM International Conference on Information and Knowledge Management}{Nov 01--05, 2021}{Queensland, Australia}
\acmBooktitle{CIKM '21: ACM International Conference on Information and Knowledge Management,
  Nov 01--05, 2021, Queensland, Australia}
\acmPrice{15.00}
\acmISBN{978-1-4503-8446-9/21/11}

\usepackage{bm}

\begin{document}

\fancyhead{}
\title{How Powerful is Graph Convolution for Recommendation?}
 \author{Yifei Shen, Yongji Wu, Yao Zhang, Caihua Shan, Jun Zhang, Khaled B. Letaief, Dongsheng Li
 }
 \affiliation{%
 	\institution{HKUST, Duke University, Fudan University, Microsoft Research Asia \\
 		yshenaw@connect.ust.hk, wuyongji317@gmail.com, \{yaozhang,dongshengli\}@fudan.edu.cn, \\ \{eejzhang, eekhaled\}@ust.hk, 	caihuashan@microsoft.com
 }
\country{}
 }


\begin{abstract}
  Graph convolutional networks (GCNs) have recently enabled a popular class of algorithms for collaborative filtering (CF). Nevertheless, the theoretical underpinnings of their empirical successes remain elusive. In this paper, we endeavor to obtain a better understanding of GCN-based CF methods via the lens of graph signal processing. By identifying the critical role of smoothness, a key concept in graph signal processing, we develop a unified graph convolution-based framework for CF. We prove that many existing CF methods are special cases of this framework, including the neighborhood-based methods, low-rank matrix factorization, linear auto-encoders, and LightGCN, corresponding to different low-pass filters. Based on our framework, we then present a simple and computationally efficient CF baseline, which we shall refer to as Graph Filter based Collaborative Filtering (GF-CF). Given an implicit feedback matrix, GF-CF can be obtained in a closed form instead of expensive training with back-propagation. Experiments will show that GF-CF achieves competitive or better performance against deep learning-based methods on three well-known datasets, notably with a $70\%$ performance gain over LightGCN on the Amazon-book dataset.
\end{abstract}




\maketitle

\section{Introduction}
Recommender systems have achieved great successes in many businesses, e.g., for product recommendation on Amazon~\cite{amazon} and playlist generation on Youtube~\cite{covington2016deep}, etc. As the algorithmic effectiveness will have a direct impact on the commercial success, building a good recommendation engine, especially via collaborative filtering (CF), remains an active research area, with consistent innovations in both conventional methods~\cite{rendle2012bpr,koren2008factorization,chen2021scalable} and recently emerged deep learning approaches~\cite{liang2018variational,he2017neural,he2020lightgcn}. 

Over the past decade, we have witnessed great progress in CF algorithms. Model-based methods largely resort to low-dimensional structures in high-dimensional data \cite{Wright-Ma-2021}, e.g., low-rank matrix factorization \cite{hu2008collaborative,pan2008one,chen2021scalable,mrma}  and autoencoders \cite{wu2016collaborative,liang2018variational,steck2020autoencoders}. On the other hand, neighborhood-based methods \cite{aiolli2013efficient,verstrepen2014unifying} achieve competitive performance based on simple similarity measures, e.g., the cosine similarity between items. Furthermore, these two types of methods can be incorporated together to improve the performance, e.g., SVD++ \cite{koren2008factorization}. From the graph perspective, the neighborhood-based methods and SVD++ effectively exploit the one-hop information in the user-item interaction graph. 

To take advantage of the rich multi-hop neighborhood information, graph convolutional networks (GCNs), e.g., GC-MC \cite{berg2017graph}, NGCF \cite{wang2019neural}, LightGCN \cite{he2020lightgcn}, have been recently proposed  and become state-of-the-art methods for CF. NGCF \cite{wang2019neural} was inspired by the GCNs developed for attribute graphs \cite{kipf2016semi}, and it inherits the key ingredients from GCNs, including initial embeddings, feature transformation, neighborhood aggregation, and nonlinear activation. As the graphs in CF tasks are non-attributed, these operations may not be necessary \cite{he2020lightgcn}. Therefore, in LightGCN \cite{he2020lightgcn}, only the most important components, i.e., trained initial embeddings and graph convolution, are preserved. Removing the unnecessary components leads to easier training and better generalization \cite{xu2021optimization,xu2020neural}, and thus LightGCN significantly outperforms NGCF in both accuracy and efficiency. While these empirical studies have produced promising results, the underlying reasons for the effectiveness of these methods remain elusive. From the theoretical perspective, an intriguing question is \emph{what plays an essential role in the success of GCN-based methods for CF}. From the practical perspective, it is interesting to investigate \emph{to what extent we can reduce the training cost while effectively exploiting the rich information of the user-item interaction graph}.

This paper endeavors to obtain a better understanding of GCN-based methods and develop a unified framework based on graph convolution that incorporates classic methods. In particular, we identify the importance of a key concept in graph signal processing in developing CF algorithms, namely, \emph{smoothness}. Conceptually, if a user interacted with an item, then  their embeddings should be similar. In graph signal processing, the similarity between the embeddings of the interacted user-item pair defines the smoothness of the embedding. Meanwhile, low-pass filters on graphs, e.g., the light convolution in LightGCN \cite{he2020lightgcn}, are used to promote the smoothness of graph signals. We will therefore argue that it is the smoothness of the embeddings and the low-pass filtering that play a pivotal role in GCN-based methods. By theoretical analysis and experiments, we will show that the performance of \emph{untrained} LightGCN is competitive to a trained one when the embedding dimension is sufficiently large, due to the smoothing effect of the light convolution. Inspired by this finding, we derive a closed-form solution for the untrained LightGCN with Infinitely Dimensional Embedding (LGCN-IDE). It is shown that LGCN-IDE outperforms LightGCN by more than $40\%$ on the Amazon-book dataset.

Motivated by its simplicity and effectiveness, we extend LGCN-IDE to incorporate general low-pass filters, which form a unified framework for CF. Surprisingly, it is proved that the neighborhood-based methods \cite{aiolli2013efficient}, low-rank matrix factorization \cite{chen2021scalable}, and linear auto-encoders \cite{steck2020autoencoders} are all special cases of this framework with various classic low-pass filters. This finding verifies the effectiveness of graph convolution with low-pass filters for CF. We further present a simple and computationally efficient CF method, which is an integration of linear filters and an ideal low-pass filter. Given an implicit feedback matrix, our proposed method has a closed-form solution and as such it would not require expensive training. More importantly and despite of its simplicity, the proposed method achieves competitive or better performance compared with deep learning methods.

To summarize, this work has made the following contributions.
\begin{enumerate}
    \item By identifying the critical role of the smoothness and low-pass filtering, we provide a novel perspective to understand the algorithms for CF.
    \item Using both theoretical justification and experiments, we show that the \emph{untrained} LightGCN can achieve competitive performance as a trained one when the embedding dimension is sufficiently large. We further derive a closed-form solution for untrained LightGCN with infinitely dimensional embedding.
    \item Built upon the closed-form solution, we develop a general graph filter-based framework for CF. We prove that the neighborhood-based methods, linear auto-encoders, and low-rank matrix factorization are special cases of this framework, corresponding to various classic low-pass filters.
    \item We present a simple and computationally efficient method, named GF-CF. With a small fraction of training time, GF-CF achieves competitive or higher performance compared with the state-of-the-art deep learning methods on three well-known datasets.
\end{enumerate}
The rest of this paper is organized as follows. Section \ref{sec:pre} introduces some preliminaries for the rest of this paper. Section \ref{sec:smooth} demonstrates the importance of smoothness. Section \ref{sec:framework} provides the details of our method. Section \ref{sec:exp} presents the experimental results. Section \ref{sec:literature} discusses the related works in CF and GCNs. Finally, we conclude this work in Section \ref{sec:con}. The code to reproduce the experiments is available at \url{https://github.com/yshenaw/GF_CF}.

\section{Preliminaries}
\label{sec:pre}
\subsection{Notations and Terminology}
This subsection presents some useful notations and definitions. We first define user set $\mathcal{U}$ and item set $\mathcal{I}$. As in \cite{he2020lightgcn}, this paper considers the recommendation problem with implicit feedback. The implicit feedback matrix $\bm{R} \in \{0,1\}^{|\mathcal{U}| \times |\mathcal{I}|}$ is defined as  follows:
\begin{align*}
    \bm{R}_{u,i} = \left\{
             \begin{array}{lr}
             1, & \text{if $(u,i)$ interaction is observed,}   \\
             0, & \text{otherwise},
             \end{array}
\right.
\end{align*}
and $\bm{r}_u$ denotes the $u$-th row of $\bm{R}$.

The adjacency matrix of the user-item interaction graph is given by 
\begin{align} \label{graph:bi}
    \bm{A} = \begin{bmatrix}
    \bm{0} & \bm{R} \\
    \bm{R}^T & \bm{0} 
    \end{bmatrix}.
\end{align}
In this bipartite graph, we denote the neighbors of node $k$ as $\mathcal{N}_k$, and its cardinality as $N_k = |\mathcal{N}_k|$.

We denote the all one column vector of any dimension as $\bm{1}$, and degree matrices as $\bm{D}_U = \text{Diag}(\bm{R} \cdot \bm{1})$ and $\bm{D}_I = \text{Diag}(\bm{1}^T \bm{R})$. The normalized rating matrix is denoted as 
\begin{align*}
    \tilde{\bm{R}} = \bm{D}_U^{-\frac{1}{2}} \bm{R} \bm{D}_I^{-\frac{1}{2}},
\end{align*}
with $\tilde{\bm{r}}_u$ as the $u$-th row of $\tilde{\bm{R}}$. Similarly, the normalized user-item adjacency matrix is given by
\begin{align*}
    \tilde{\bm{A}} = \begin{bmatrix}
    \bm{0} & \tilde{\bm{R}} \\
    \tilde{\bm{R}}^T & \bm{0} 
    \end{bmatrix}.
\end{align*}

We also define the item-item normalized adjacency matrix as 
\begin{align*}
    \tilde{\bm{P} } = \tilde{\bm{R}}^T\tilde{\bm{R}}.
\end{align*}

We then define an important concept, namely, Stiefel manifold, which can help to connect low-rank matrix factorization and GCN-based methods in Section \ref{sec:related}.

\begin{definition} (Stiefel manifold) The Stiefel manifold $\text{St}(n,m)$ is defined as the subspace of orthonormal N-frames in $\mathbb{R}^n$, namely,
\begin{equation}
\text{St}(n,m) = \{\bm{\Gamma} \in \mathbb{R}^{n\times m}: \bm{\Gamma}^T\bm{\Gamma} = \bm{I}  \}
\end{equation}
where $\bm{I}$ is the identity matrix.
\end{definition}
\subsection{Graph Signal Processing}
In this subsection, we introduce basic concepts of graph signal processing \cite{dong2020graph,ramakrishna2020user}. We consider a weighted undirected graph $\mathcal{G} = (\mathcal{V},\mathcal{E})$ with $n$ nodes where $\mathcal{V}$ and $\mathcal{E}$ denote the vertex set and edge set, respectively. The graph can be represented as an adjacency matrix $\bm{A} \in \mathbb{R}^{n \times n}$. A graph signal is defined as a function $x:\mathcal{V} \rightarrow \mathbb{R}$ and it can be represented as a $n$-dimensional vector $\bm{x} = [x(i)]_{i \in \mathcal{V}}$. For a graph signal, the derivative is defined as $(\nabla \bm{x})_{i,j} = \sqrt{A_{i,j} } (x_i - x_j)$.

The smoothness of a graph signal can be measured by the \emph{graph quadratic form}, which is the squared norm of the graph derivative as defined below:
\begin{align*}
    S(\bm{x}) = \frac{1}{2} \|\nabla \bm{x}\|_F^2 = \sum_{i,j} A_{i,j} (x_i - x_j)^2 = \bm{x}^T \bm{L} \bm{x}.
\end{align*}
Here, $\bm{L} = \bm{D} - \bm{A}$ is the graph Laplacian matrix\footnote{The graph Laplacian matrix can also be defined by some normalized version of $\bm{D} - \bm{A}$, e.g., $\tilde{\bm{L}} = \bm{I} - \tilde{\bm{A}}$.}. A smaller $\frac{S(\bm{x})}{\|\bm{x}\|_2}$ indicates smoother signals. 

In many applications, a graph signal is often described in a vector form $x: \mathcal{V} \rightarrow \mathbb{R}^d$ and its smoothness can be written as follows:
\begin{align} \label{eq:vec_smooth}
    S_2(\bm{x}) = \sum_{i,j} A_{i,j} \|\bm{x}_i - \bm{x}_j\|_2^2.
\end{align}

As $\bm{L}$ is real and symmetric, its eigendecomposition is given by $\bm{L} = \bm{U} \bm{\Lambda} \bm{U}^T$ where $\bm{\Lambda} = \text{Diag}(\lambda_1, \cdots, \lambda_n)$, $\lambda_1 \leq \lambda_2 \leq \dots \leq \lambda_n$, and $\bm{U} = [\bm{u}_1, \cdots, \bm{u}_n]$ with $\bm{u}_i \in \mathbb{R}^n$ being the eigenvector for eigenvalue $\lambda_i$.

Next we discuss the frequency of the graph signal and define Fourier Transform on graphs. Intuitively, the graph signal has a  higher frequency if it is more 
oscillatory and not smooth. As $\lambda_1$ is the smallest eigenvalue, for any graph signal $\bm{x} \in \mathbb{R}^n$, we have $\frac{S(\bm{x})}{\|\bm{x}\|_2} \geq \frac{S(\bm{u}_1)}{\|\bm{u}_1\|_2}$. Thus, the eigenvector with a smaller eigenvalue corresponds to a lower frequency signal component. We can define the Graph Fourier Transform (GFT) basis as the eigenvector matrix $\bm{U}$ and we call $\hat{\bm{x}} = \bm{U}^T \bm{x}$ as the GFT of the graph signal $\bm{x}$. Similar to the Fourier transform, GFT is a linear orthogonal transform and its inverse transform is given by $\bm{x} = \bm{U} \hat{\bm{x}}$. GFT enables us to define graph filters and graph convolution. 
\begin{definition} (Graph Filter)
Given a graph Laplacian matrix, as well as its eigenvectors and eigenvalues, then the graph filter $\mathcal{H}(\bm{L})$ is defined as follows:
\begin{align*}
    \mathcal{H}(\bm{L}) = \bm{U} \text{Diag}(h(\lambda_1), \cdots, h(\lambda_n)) \bm{U}^T.
\end{align*}
where $h(\cdot)$ is the filter defined on the eigenvalues. 
\end{definition}

\begin{definition} (Graph Convolution) The graph convolution of a input signal $\bm{x}$ and the filter  $\mathcal{H}(\bm{L})$ is defined as follows:
\begin{align*}
    \bm{y} = \mathcal{H}(\bm{L}) \bm{x} = \bm{U} \text{Diag}(h(\lambda_1), \cdots, h(\lambda_n)) \bm{U}^T\bm{x}.
\end{align*}
\end{definition}
Similar to the definition of convolution in classic signal processing, the graph signal is transformed by GFT $\bm{U}^T$, multiplied by a filter $h(\cdot)$, and transformed back by inverse GFT $\bm{U}$. In the context of CF, the graph signal is often the observed ratings for a given user \cite{chen2021scalable}, or the initial embeddings of users/items \cite{wang2019neural,he2020lightgcn}.

In the signal processing literature, the signal is often smooth and with low frequency, and the noise is often non-smooth and with a high frequency. One important class of filters is the \emph{low-pass} filters, which promotes smoothness of graph signals for denoising. The graph low-pass filters are defined as follows. 
\begin{definition} (Low-pass Filter) \label{def:lpf}
    For $k = 1, \cdots, n-1$, we define the ratio
    \begin{align} \label{eq:low-pass}
        \eta_k := \frac{\max\{|h(\lambda_{k+1})|, \cdots, |h(\lambda_{n})| \} }{\min\{ |h(\lambda_{1})|, \cdots, |h(\lambda_{k})| \} }.
    \end{align}
    The graph filter $\mathcal{H}(\bm{L})$ is $k$-low-pass if and only if the low-pass ratio satisfies $\eta_k \in [0,1)$.
\end{definition}
The low-pass ratio defines how much of the high-frequency component of a signal is allowed to pass compared to the low-frequency components. If $\eta_k < 1$, then the filter passes low-frequency signals and is  called a low-pass filter. We here list some important low-pass filters and will connect these filters with classic methods for recommendation in Section \ref{sec:related}.
\paragraph{Linear Filter} The linear filter is given by 
\begin{align} \label{fil:linear}
    h(\lambda_i) = \sum_{k=0}^{K} \alpha_k \lambda_i^k,
\end{align}
where $\alpha_k$ is the filter's coefficient. It is called \emph{linear} due to its similarity with linear time invariant filters in classic signal processing. We will show that this filter corresponds to LightGCN and neighborhood-based methods.

\paragraph{Ideal Low-pass Filter} 
The ideal low-pass filter has a cut-off frequency $\bar{\lambda}$. The filter is defined as 
\begin{align} \label{fil:low}
    h(\lambda_i) = \left\{
             \begin{array}{lr}
             1, & \text{if } \lambda_i \leq \bar{\lambda}  \\
             0, & \text{otherwise}.
             \end{array}
\right.
\end{align}
It is called \emph{ideal} as the high-frequency signals are ideally cut off with no leakage. We will show that this filter corresponds to the low-rank matrix factorization method. 

\paragraph{Opinion Dynamics} The opinion dynamics are a \emph{graph diffusion} process, which is a GF-AR(1) model \cite{friedkin2011formal}: $\bm{y}_{t+1} = (1 - \beta) (\bm{I} - \alpha \bm{L}) \bm{y}_t + \beta \bm{x}_t$. The steady state opinions are given by $\bm{y} = \lim_{t \rightarrow \infty} \bm{y}_t  = (\bm{I} + \tilde{\alpha} \bm{L})^{-1} \bm{x} = \mathcal{H}(\bm{L}) \bm{x}$ where $\tilde{\alpha} = \beta (1 - \alpha) \alpha$. Thus, the corresponding graph filter is
\begin{align} \label{fil:opi}
    h(\lambda_i) = \frac{1}{1 + \tilde{\alpha} \lambda_i}.
\end{align}
In opinion dynamics, the matrix inverse or eigenvalue decomposition is required. Thus, applying this filter introduces a high memory cost. We will show that it is closely related to the linear auto-encoder method.
\begin{figure*}[htbp]
    \centering
    \subfigure[Recall on Gowalla dataset.]
    {
        \includegraphics[width=0.48\columnwidth]{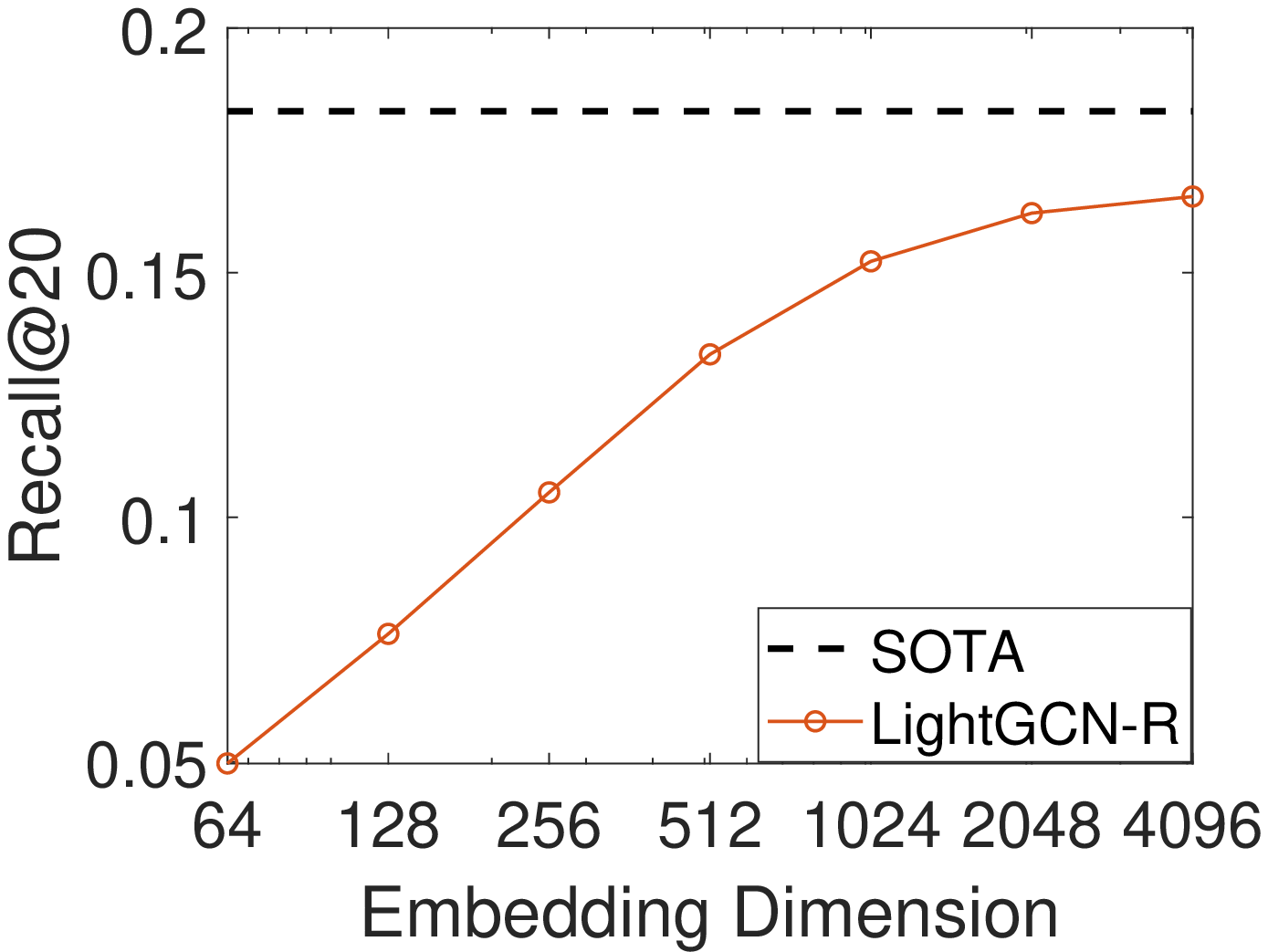}

    }
    \subfigure[NDCG on Gowalla dataset.]
    {
        \includegraphics[width=0.48\columnwidth]{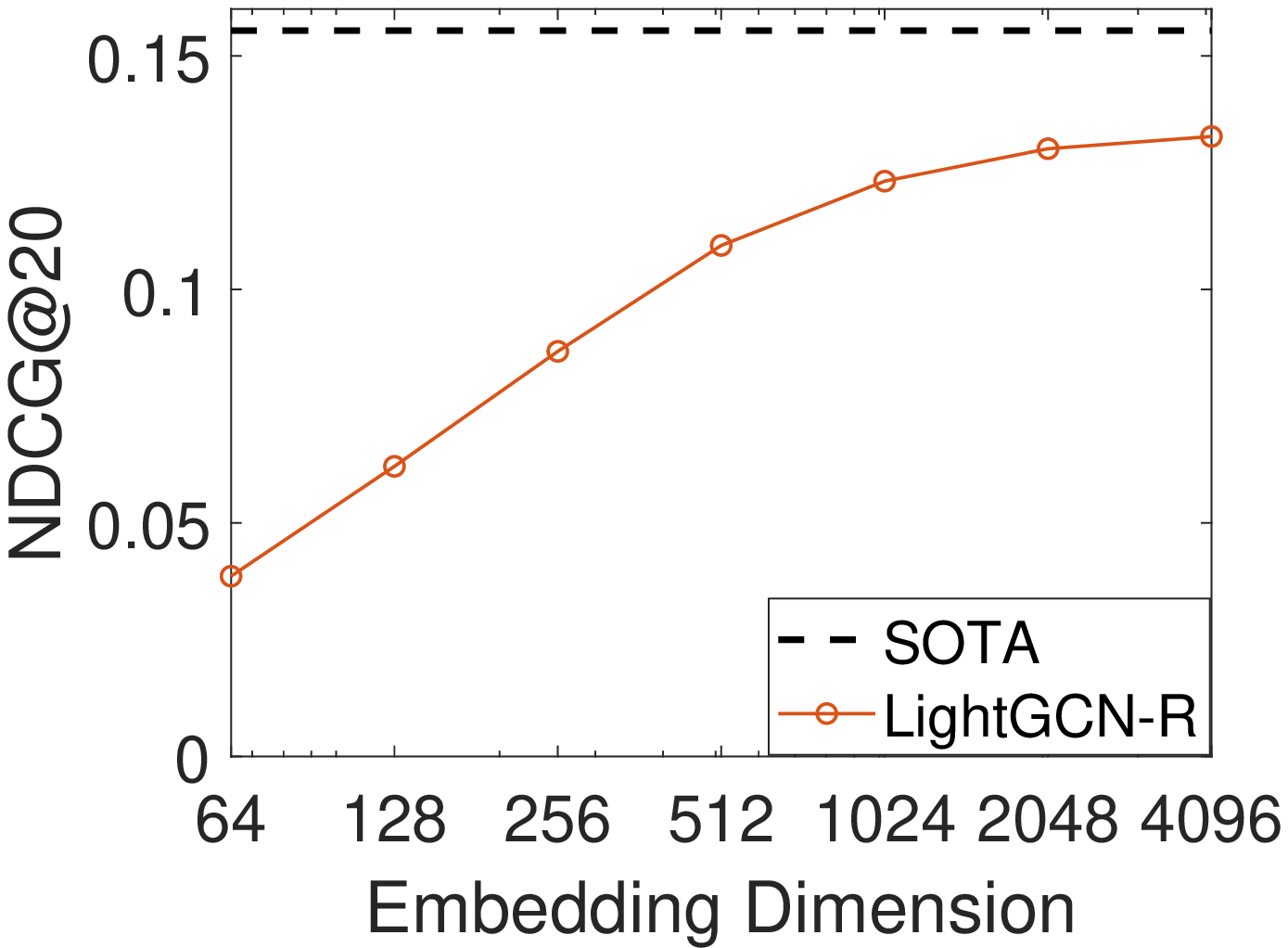}

    }
    \subfigure[Recall on Amazon-book dataset.]
    {
        \includegraphics[width=0.48\columnwidth]{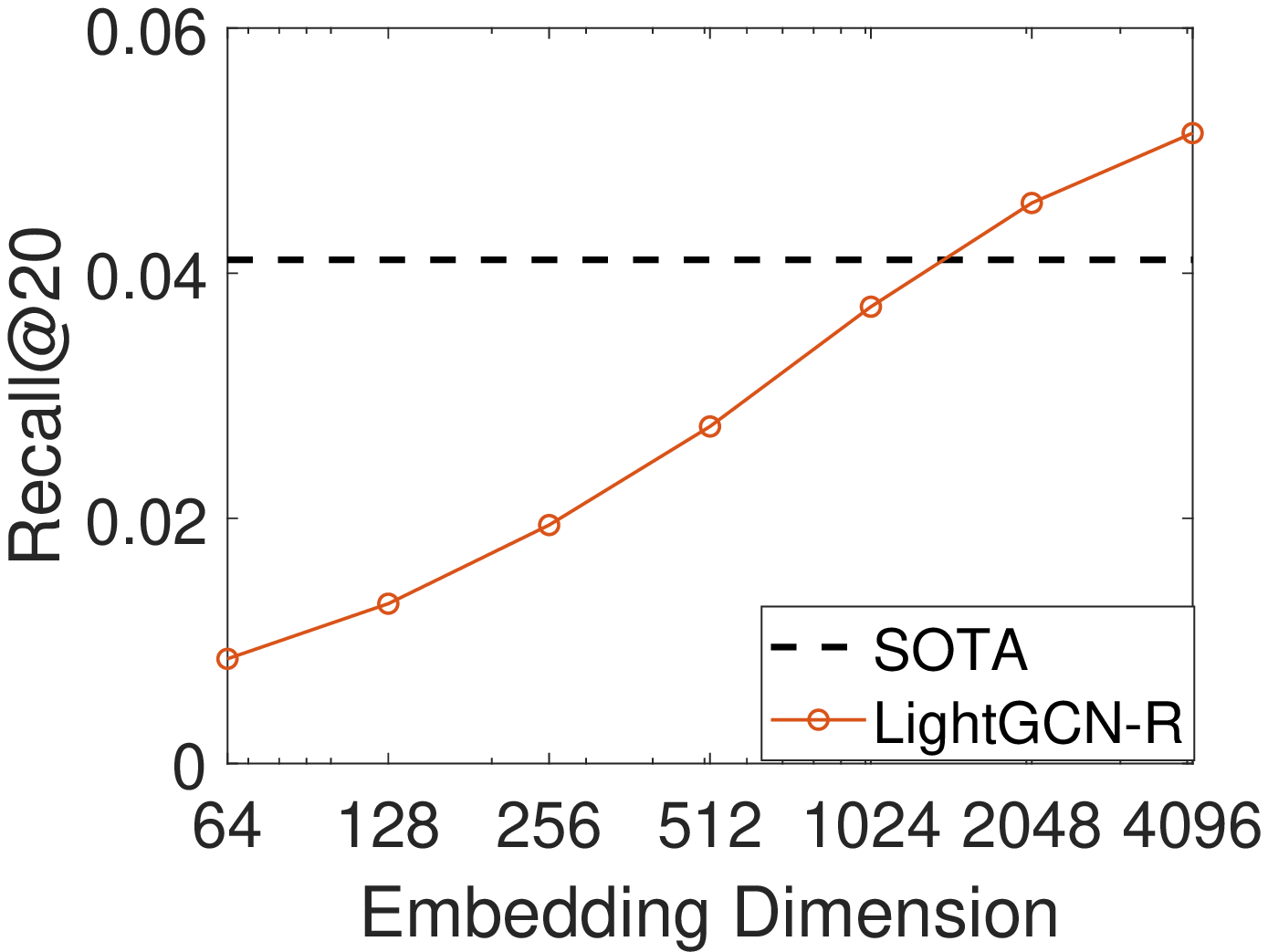}

    }
    \subfigure[NDCG on Amazon-book dataset.]
    {
        \includegraphics[width=0.48\columnwidth]{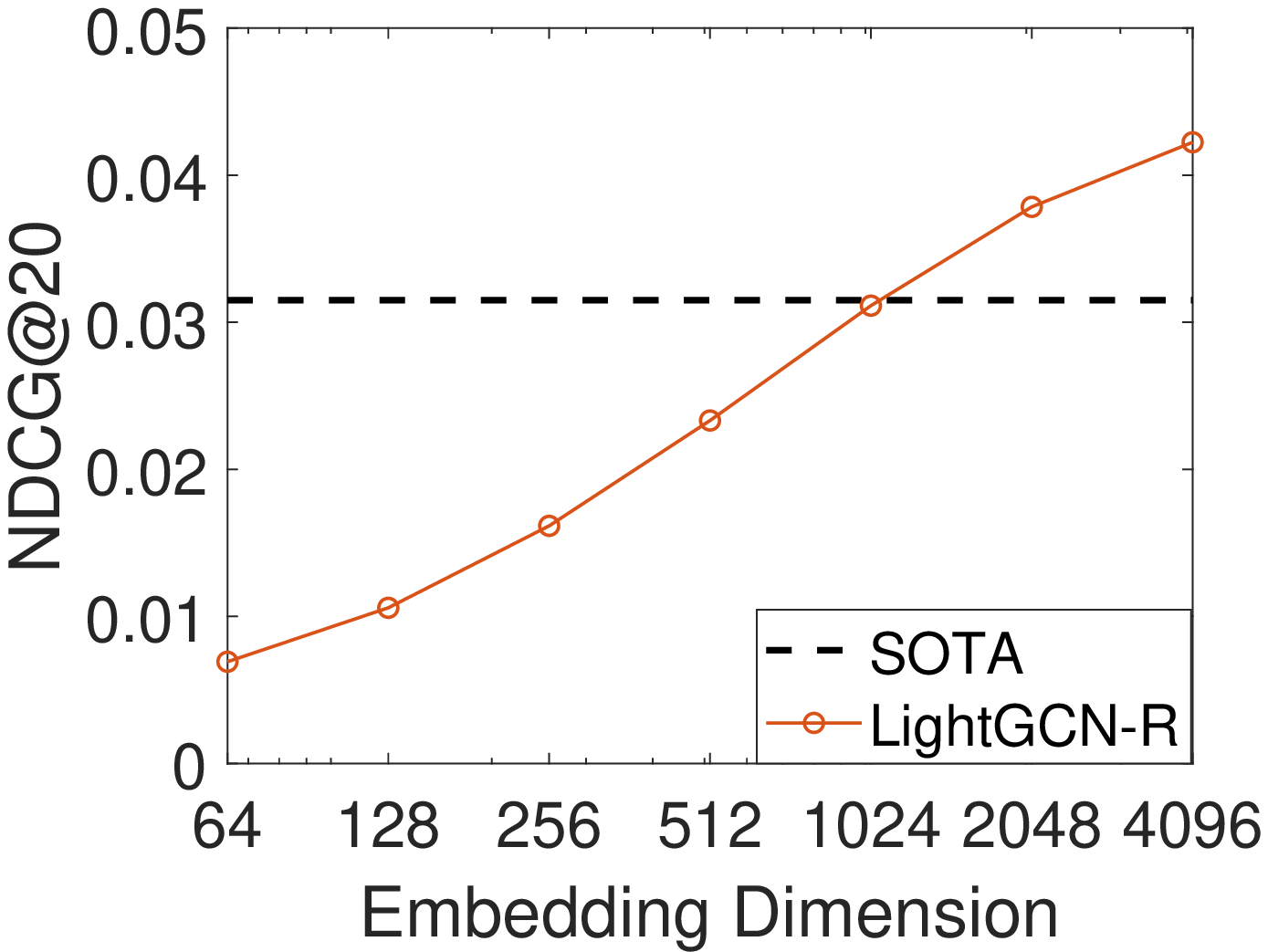}

    }
        \caption{
        Performance of \emph{untrained} LightGCN versus SOTA, where the SOTA line is LightGCN's performance reported in \cite{he2020lightgcn} and LightGCN-R denotes the untrained LightGCN with different embedding dimensions. }
    \label{fig:untrained}
\end{figure*}

\subsection{LightGCN Brief} \label{sec:lightgcn}
LightGCN \cite{he2020lightgcn} is a state-of-the-art GCN-based method in CF. In this paper, LightGCN will be used as the vehicle to elaborate our theory and adopted as the main baseline for performance comparison. 

LightGCN leverages the user-item interaction graph to propagate the embedding as follows:
\begin{align*}
    \bm{E}^{(k+1)} = \tilde{\bm{A}}\bm{E}^{(k)},
\end{align*}
where $\bm{E}^{(0)} \in \mathbb{R}^{(|\mathcal{U}|+|\mathcal{I}|) \times d}$ is the learnable initial embedding matrix of users and items. For a $K$-layer LightGCN, the final embeddings can be computed as follows:
\begin{align} \label{LGCN:final}
    \bm{E} &= \alpha_0 \bm{E}^{(0)} + \alpha_1 \bm{E}^{(1)} + \cdots + \alpha_K \bm{E}^{(K)} \notag\\
    &= \alpha_0 \bm{E}^{(0)} +  \alpha_1 \tilde{\bm{A}}\bm{E}^{(0)} + \cdots + \alpha_K \tilde{\bm{A}}^{K}\bm{E}^{(0)}.
\end{align}
The model prediction is defined as the inner product of the user's and item's final representation $y_{ui} = \bm{e}_u^T \bm{e}_i$, where $\bm{e}_u$ and $\bm{e}_i$ are the corresponding rows of $\bm{E}$.

To optimize LightGCN, the Bayesian personalized ranking (BPR) loss \cite{rendle2012bpr} is adopted:
\begin{align} \label{loss:bpr}
    l_{\text{BPR}} = - \sum_{u = 1}^{|\mathcal{U}|} \sum_{i \in \mathcal{N}_u} \sum_{j \notin \mathcal{N}_u } \log \sigma(\bm{e}_u^T \bm{e}_i -  \bm{e}_u^T \bm{e}_j).
\end{align}

\section{On the Importance of Smoothness and Low-pass Filtering}\label{sec:smooth}
In this section, we identify the importance of smoothness and low-pass filters in CF, by using the light convolution in LightGCN as a specific example.

The embeddings play an essential role in CF while smoothness is a key concept in graph signal processing. We observe that there are strong connections between the good embeddings and their smoothness on the graph. We consider the dot product based embedding model. Specifically, let $\bm{e}_u$ denote the embedding for the $u$-th user and $\bm{e}_i$ denote the embedding for the $i$-th item. The predicted score for the $u$-th user and $i$-th item is defined by the dot product $\bm{e}_u^T \bm{e}_i$. If $\bm{R}_{u,i} = 1$, we should promote the similarities between $\bm{e}_u$ and $\bm{e}_i$. By the definition of smoothness of graph signals \eqref{eq:vec_smooth}, if $\bm{e}_u$ and $\bm{e}_i$ are similar for connected user-item pairs, the embeddings are smooth signals on the graph. Consequently, optimizing loss functions, e.g., BPR loss \eqref{loss:bpr}, and enhancing the smoothness of the embeddings share the same goal: promoting the similarity between $\bm{e}_u$ and $\bm{e}_i$ for $\bm{R}_{u,i} = 1$. 

The above discussion provides a qualitative intuition for the role of smoothness in embeddings. We will now analyze the linear filter in LightGCN to obtain quantitative results. The LightGCN consists of two components: the initial embedding and a linear filter. If the \emph{untrained} LightGCN achieves good performance, it must be the linear filter playing the essential role as the initial embedding is random. The next proposition shows that untrained LightGCN will have a low BPR loss under certain conditions.

\begin{theorem} \label{prop:untrain}
    Denote $N_{\max} = \max_i N_i$ and $N_{\min} = \min_i N_i$ where $i \in \mathcal{U} \cup \mathcal{I}$. If $\bm{E}^{(0)} \in \mathbb{R}^{(|\mathcal{I}| + |\mathcal{U}|) \times d}$ follows an i.i.d. uniform distribution over the unit sphere with
    \begin{align} \label{eq:dims}
        d > \frac{ C N_{\max}^{3} \log(|\mathcal{I}| + |\mathcal{U}|)}{N_{\min}},
    \end{align}
    then for a one-layer untrained LightGCN, we have
    \begin{align} \label{eq:rnd_bpr}
        \mathbb{P} \left(\left\{\bm{e}^{(1)T}_u \bm{e}^{(0)}_i   > \bm{e}_u^{(1)T} \bm{e}^{(0)}_j | (u,i) \in \mathcal{S}_1, (u,j) \in \mathcal{S}_2\right\} \right) \geq 3/4,
    \end{align}
    where $C$ is an absolute constant, $\mathcal{S}_1 = \{(u,i)|\bm{R}_{u,i}=1\}$, $\mathcal{S}_2 = \{(u,j)|\bm{R}_{u,j}=0\}$.
\end{theorem}

\begin{remark} (Interpretations of Theorem \ref{prop:untrain}) Equation \eqref{eq:rnd_bpr} implies a low BPR loss as the predicted score of any positive pair is larger than that of any negative pair. Due to the smoothing effect of light convolution (linear filters), the final embeddings between interacted pairs are similar even if the initial embeddings are random. In Equation \eqref{eq:dims}, $N_{\max}$ and $N_{\min}$ are adopted for a worst-case analysis, and in practice, we can replace them with the average degree. Equation \eqref{eq:dims} shows that the required embedding dimension of untrained LightGCN grows with the dataset density, which implies untrained LightGCN is more effective on sparse datasets. For the Gaussian initialization adopted in \cite{he2020lightgcn}, the results are similar as high dimensional Gaussian random vectors concentrate around a sphere (refer to Section 3.1 in \cite{vershynin2018high}). The probability $3/4$ can be improved to any probability approaches arbitrarily close to $1$. The high-level reason for untrained LightGCN performing well is that the information contained in the rating matrix and the graph are identical. The BPR loss is adopted for exploiting the information in rating matrix while the low-pass filters are to exploit information in the graph. Thus, a proper use of low-pass filters can accelerate the training or even avoid the training. Interestingly, some recent works also reveal that infinitely wide random CNNs achieve better performance than trained ones \cite{arora2019harnessing}.
\end{remark}

Based on Theorem \ref{prop:untrain}, we argue that the performance of \emph{untrained} LightGCN improves with the embedding dimension and it should be competitive to a trained one when the embedding dimension is sufficiently large. 

To verify this argument, we follow the experiment settings in \cite{he2020lightgcn} and conduct the experiments for a $3$-layer \emph{untrained} LightGCN. The initial embeddings $\bm{E}^{(0)}$ is initialized following an i.i.d. Gaussian distribution $\mathcal{N}(0,0.1)$ as in the original paper. Once the model is initialized, we do not train it but simply compute the user/item embeddings using Equation \eqref{LGCN:final} and then directly test it on the test dataset. We use two sparse datasets, i.e., Gowalla and Amazon-book. The test performance versus the embedding dimension is shown in Fig. \ref{fig:untrained}. As the training/test splitting of two datasets is identical to \cite{he2020lightgcn}, we regard the LightGCN's performance reported in \cite{he2020lightgcn} as the state-of-the-art. We will also compare to LightGCN with large embedding dimensions in Table \ref{tab:time}. The experiments agree with our theory well. As the linear filter is to promote the smoothness, it demonstrates the crucial role of smoothness in CF.

However, the untrained LightGCN is not a practical algorithm for recommendation as the large embedding dimension leads to an expensive memory cost and inference time. Fortunately, the untrained LightGCN with infinitely dimensional embedding has a closed-form solution for predicted scores, as shown in the next theorem. 

\begin{table*}[!htb]
	
	\selectfont  
	\centering
	\caption{Classic methods versus their corresponding graph filters and spatial GCNs.} 
	
	\resizebox{1\textwidth}{!}{
		\begin{tabular}{|c|c|c|c|c|c|c|}  
			\hline  
			Method &   Low-rank Factorization \cite{chen2021scalable} & Linear Auto-encoder \cite{steck2020autoencoders} & Neighborhood-based \cite{aiolli2013efficient} & LGCN-IDE \eqref{LGCN-INF} \\ \hline
			Input Signal $\bar{\bm{r}}_u$  & $\bm{D}^{-\frac{1}{2}}_I\bm{r}_u$ & $\tilde{\bm{r} }_u$ &   $\bm{r}_u$ & $\tilde{\bm{r} }_u$ \\ \hline
			Graph Filter  & $h(\lambda_i) = \bm{1}_{i \leq d}$ & $h(\lambda_i) = \frac{ 1 - \lambda_i}{ 1 + \mu - \lambda_i}$ & $h(\lambda_i) = 1 - \lambda_i$ &  $h(\lambda_i) = \sum_{k=0}^{K-1} \beta_k (1 - \lambda_i)^k$  \\ \hline
			Corresponding 
			Spatial GCN  & \tabincell{c}{Infinite-layer spatial  GCN with\\ convolutional normalization \eqref{eq:mf_1} \eqref{eq:mf_2}} & \tabincell{c}{Infinite-layer spatial GCN \\ with layer combination \eqref{eq:ae}} & Single-layer spatial GCN & \tabincell{c}{Multi-layer spatial GCN \\ with layer combination } \\\hline
	\end{tabular}}
	\label{tab:unify}
\end{table*}

\begin{theorem} \label{thm:inf_lgcn}
Consider an untrained LightGCN with \\$\bm{E}^{(0)} \in \mathbb{R}^{(|\mathcal{I}|+|\mathcal{U}|)\times d}$ following an i.i.d. distribution with zero mean and non-zero variance. As $d \rightarrow \infty$, the predicted score of the untrained LightGCN follows
\begin{align} \label{LGCN-INF}
    \bm{s}_u = \sum_{k=0}^{K-1} \beta_k \tilde{\bm{r}}_u (\tilde{\bm{R}}^T\tilde{\bm{R}})^k.
\end{align}
where $\beta_k$ are constants depending on $[\alpha_k]_{k=0,\cdots,K}$ in \eqref{LGCN:final}.
\end{theorem}

\begin{remark} (Interpretations of $\tilde{\bm{R}}^T\tilde{\bm{R}}$) As shown in Theorem \ref{thm:inf_lgcn}, the gram matrix $\tilde{\bm{R}}^T\tilde{\bm{R}}$ plays a pivotal role. For sparse binary data, $(\bm{D}_I^{-\frac{1}{2}}\bm{R}^T\bm{R}\bm{D}_I^{-\frac{1}{2}})_{ij}$ defines the cosine similarity between item $i$ and item $j$ \cite{aiolli2013efficient}. Likewise, $(\tilde{\bm{R}}^T\tilde{\bm{R}})_{ij}$ provides a similarity measure between item $i$ and item $j$. Directly using the gram matrix as the item-item similarity results in the neighborhood-based method \cite{aiolli2013efficient}, which was the winner of Millions of Song Competition\footnote{The competition is at https://www.kaggle.com/c/msdchallenge}. The similarity between the neighborhood-based method and LightGCN is not surprising as LightGCN is based on the neighborhood propagation. As LightGCN consists of multi-hop propagation, the term $\tilde{\bm{R}}^T\tilde{\bm{R}}$ appears as polynomials. From the graph signal processing perspective, it is a linear filter, which is low-pass.
\end{remark}

We call \eqref{LGCN-INF} as LightGCN with Infinitely Dimensional Embedding (LGCN-IDE). The performance of LGCN-IDE is shown in Table \ref{tab:overall}. Remarkably, we see that on Amazon-book dataset, it outperforms the performance of LightGCN reported in \cite{he2020lightgcn} by more than $40\%$ under exactly the same training/test data splits.
 
\section{A Unified Framework}
\label{sec:framework}
In this section, we first extend LGCN-IDE to incorporate general low-pass filters, which form a unified framework. Then we prove that this framework unifies the neighborhood-based approaches, low-rank matrix factorization, linear auto-encoders, and linear graph convolutional networks, where different methods correspond to different low pass-filters. Finally, we present a simple yet effective algorithm for CF. 

\subsection{A Unified Graph Low-pass Filter Based Framework} 
In this subsection, we extend \eqref{LGCN-INF} to incorporate general graph filters. To simplify the notations, we denote $\tilde{\bm{P}} = \tilde{\bm{R} }^T\tilde{\bm{R} }$ in the remaining of the article. Note that $\tilde{\bm{P}}$ can also be seen as a normalized adjacency matrix for an item-to-item graph, whose eigenvalues are between $0$ and $1$. 

\begin{theorem} \label{thm:P}
    Let $\lambda_1 \geq \cdots \geq \lambda_{|\mathcal{I}|}$ be the eigenvalues of $\tilde{\bm{P}}$, then
    \begin{align*}
        0 \leq \lambda_{|\mathcal{I}|} \leq \cdots \leq \lambda_1 \leq 1.
    \end{align*}
\end{theorem}
The graph Laplacian of the item-to-item graph is defined as $\tilde{\bm{L}} = \bm{I} - \tilde{\bm{P}}$. In this way, we can apply graph signal processing to the item-to-item graph. Next we elaborate our unified framework, which is an extension of \eqref{LGCN-INF} with general graph filters. We consider an input graph signal $\bar{\bm{r} }_u$, which is some transformation of the users' observed ratings $\bm{r}_u$. Then a low-pass filter is applied to the graph signal to obtain a filtered signal. Finally, we may scale the obtained graph signal to get the final prediction scores. Denoting the eigendecomposition by $\tilde{\bm{L}} = \bm{U}\bm{\Lambda}\bm{U}^T$, the framework is given by 
\begin{align} \label{eq:framework}
    \bar{\bm{s}}_u = \bar{\bm{r}}_u \bm{U} \text{Diag}(h(\lambda_1), \cdots, h(\lambda_n)) \bm{U}^T,
\end{align}
where $\bar{\bm{s}}_u$ is the filtered predicted score, and $h(\cdot)$ is a low-pass filter. From the graph signal processing perspective, it is a graph convolution, i.e., a graph signal $\bm{r}_u \in \mathbb{R}^{|\mathcal{I}|}$ convolving with a low-pass filter $h(\cdot)$.

\subsection{Interpreting Classic Methods from Graph Signal Processing Perspective} \label{sec:related}
Interestingly, some classic works for recommendation can be interpreted as graph signal processing approaches, where the low-pass filter plays an essential role. The classic methods typically involves auto-encoder-based \cite{ning2011slim,liang2018variational,steck2020autoencoders}, matrix factorization-based \cite{rendle2012bpr,chen2021scalable}, and GCN-based ones \cite{wang2019neural,he2020lightgcn,zhang2021cope}. In this subsection, we will provide a unified view of the linear methods from the graph signal processing perspective. As the spectral convolution can be transformed into a spatial convolution in GCNs by first-order approximation \cite{kipf2016semi}, it is interesting to investigate what kind of GCNs will these classic methods induce. These GCNs induced by classic algorithms can also be seen as white-box neural networks \cite{chan2021redunet}. A test of performance for these GCNs is left for future works.

\subsubsection{Low-rank Matrix Factorization} Low-rank matrix factorization is one of the most classic algorithms for CF. Note that GFT is also a matrix factorization where the low-frequency signal components correspond to the principle components of the rating matrix. This observation allows us to connect MF and graph-based methods. We take the objective function in a recent work \cite{chen2021scalable} as an example. Denote $d$ as the embedding dimension, the model is given by
\begin{align} \label{eq:SGMC}
    \bm{U}^*, \bm{V}^* = \underset{\bm{U} \in \mathbb{R}^{|\mathcal{U}| \times d},\bm{V} \in \mathbb{R}^{|\mathcal{I}| \times d} }{\text{argmin} } \|\tilde{\bm{R}} - \bm{U}\bm{V}^T\|_F^2 \quad \text{s.t. } \bm{V}^T \bm{V} = \bm{I}.  
\end{align}
As shown in \cite{chen2021scalable}, $\bm{V}^*$ contains the smallest $K$ eigenvectors of $\tilde{\bm{L}}$ and $\bm{U}^* = \bm{R}\bm{D}_I^{\frac{1}{2}}\bm{V}$. Viewing the eigendecomposition as GFT, it can be interpreted as an ideal low-pass filter \eqref{fil:low}
\begin{align*}
    h(\lambda_i) = \bm{1}_{i \leq d}.
\end{align*}

We then turn low-rank matrix factorization into a spatial convolution fashion. This is more difficult than the conversion in GCN \cite{kipf2016semi} due to the orthogonal constraint and non-convexity of problem \eqref{eq:SGMC}. Observing that the optimal solution $\bm{V}^*$ to \eqref{eq:SGMC} is also the optimal solution to the following problem 
\begin{align} \label{eq:power_sgmc}
    \bm{V}^* = \underset{\bm{X} \in St(|\mathcal{I}|, d)}{\text{argmax}}  \| \tilde{\bm{R}}\bm{X}\|_F^2.
\end{align}
We can rewrite \eqref{eq:power_sgmc} as spatial convolution by first-order expansion like GCNs \cite{kipf2016semi}. We begin with a random $\bm{E}^{(0)} \in \mathbb{R}^{|\mathcal{I}|\times d}$, and the update rule is given by 
\begin{align}
    & \hat{\bm{E}}^{(k)} = \at{\left(\nabla_{\bm{X}} \| \tilde{\bm{R}}\bm{X}\|_F^2\right)}{\bm{X} = \bm{E}^{(k-1)}} = \tilde{\bm{P}} \bm{E}^{(k-1)}, \label{eq:mf_1}\\
    &\bm{E}^{(i)} = \underset{\bm{S} \in St(|\mathcal{I}|,d)}{\text{argmax} } \langle \bm{S}, \hat{\bm{E}}^{(k)}  \rangle \overset{(a)}{=}  \hat{\bm{E}}^{(k)} \left(\hat{\bm{E}}^{(k)T}\hat{\bm{E}}^{(k)}  \right)^{-\frac{1}{2}} \label{eq:mf_2},
\end{align}
where (a) follows Proposition 7 in \cite{journee2010generalized}. The final embeddings are given by 
\begin{align*}
    \bm{V}^* = \bm{E} = \bm{E}^{(\infty)}.
\end{align*}
Note that \eqref{eq:mf_1} is a spatial graph convolution, and \eqref{eq:mf_2} is coincidentally equivalent to convolutional normalization for CNNs \cite{liu2021convolutional} (refer to (6)-(8) in \cite{liu2021convolutional}). The convolutional normalization was proposed to accelerate the training of convolutional networks and improve robustness. From this view, the low-rank matrix factorization is equivalent to an infinite layer GCN with convolutional normalization. As the number of layers is large, it suffers from the over-smoothing issue \cite{li2018deeper}, which hurts the performance.

\subsubsection{Linear Auto-encoders} In the linear auto-encoders, e.g., EASE \cite{steck2019embarrassingly} and SLIM \cite{ning2011slim}, the predicted score vector of a user ($\bar{\bm{s}}_u$) is obtained by the dot product
\begin{align*}
    \bar{\bm{s}}_{u} = \bar{\bm{r}}_u \bm{B},
\end{align*}
where $\bm{B} \in \mathbb{R}^{|\mathcal{I}| \times |\mathcal{I}|}$ is a learnable weight matrix. The training objective is some regularized or constrained version of $\min_{\bm{B}} \ \sum_u \|\bar{\bm{r}}_u - \bar{\bm{r}}_u \bm{B}\|_2^2$ . From the graph signal processing view, it can be interpreted as a graph signal $\bar{\bm{r}}_u$ convolving with a filter $\bm{B}$, and $\bar{\bm{r}}_u = \bar{\bm{r}}_u \bm{B}$ is a steady state. This defines a graph diffusion on the corresponding graph like opinion dynamics \eqref{fil:opi}. Next, we show the equivalence between a specific version of linear auto-encoders and the graph diffusion filter.

As shown in \cite{steck2020autoencoders}, the following linear auto-encoder is able to achieve competitive performance compared with the deep ones \cite{wu2016collaborative,liang2018variational}. Specifically, we consider the following formulation in \cite{steck2020autoencoders} for simplicity:
\begin{align} \label{opt:DLAE_full}
    \underset{\bm{B}}{\text{minimize}}\quad \|\tilde{\bm{R}} - \tilde{\bm{R}} \bm{B}\|_F^2 + \mu \|\bm{B}\|_F^2.
\end{align}
As \eqref{opt:DLAE_full} is a ridge regression, we can write down the closed-form solution as 
\begin{align} \label{eq:DLAE}
    \bm{B}^* = (\tilde{\bm{P}} + \mu \bm{I})^{-1} \tilde{\bm{P}}.
\end{align}
Viewing the eigenvalue decomposition as GFT, the graph filter in \eqref{eq:DLAE} is given by
\begin{align} \label{fil:DLAE}
    h(\lambda_i) = \frac{ 1 - \lambda_i}{ 1 + \mu - \lambda_i}.
\end{align}
To understand \eqref{fil:DLAE} in the content of low-pass filters, the low-pass ratio $\eta_k$ in \eqref{eq:low-pass} is given by
\begin{align*}
    \eta_k = \frac{\frac{1 - \lambda_{k+1} }{1 + \mu - \lambda_{k+1}}}{\frac{1 - \lambda_k}{1 + \mu - \lambda_k}} &= \frac{(1 - \lambda_{k+1})(1 + \mu - \lambda_k)}{(1 - \lambda_k)(1 + \mu - \lambda_{k+1})} & \\
    &= 1 - \frac{\mu(\lambda_{k+1} - \lambda_k)}{(1 - \lambda_k)(1 + \mu - \lambda_{k+1})} < 1.
\end{align*}
The convolutional filter in \eqref{fil:DLAE} is similar to opinion dynamics and is a kind of graph diffusion filter. Like other diffusion-based methods, the memory cost of linear auto-encoder is high as we need to store matrix $\bm{B}$ in \eqref{opt:DLAE_full}.

In the literature, the Neumann series are often adopted to convert the graph diffusion into a spatial convolution \cite{klicpera2018predict,xhonneux2020continuous}. Similarly, we can use it to interpret \eqref{eq:DLAE} as spatial GCNs. For $\mu > 1$, \eqref{eq:DLAE} can be written as 
\begin{align}\label{eq:ae}
    (\tilde{\bm{P}} + \mu \bm{I})^{-1} \tilde{\bm{P}} =  \frac{1}{\mu} \left(\sum_{k=0}^{\infty} (-\mu^{-1} \tilde{\bm{P}})^k \right) \tilde{\bm{P}}.
\end{align}
From this view, the initial embedding is an identity matrix $\bm{E}^{(0)} = \bm{I} \in \mathbb{R}^{|\mathcal{I}| \times |\mathcal{I}|}$, the update of corresponding spatial convolution is given by
\begin{align*}
    \bm{E}^{(k)} = \tilde{\bm{P}} \bm{E}^{(k-1)},
\end{align*}
and the final embeddings can be obtained as
\begin{align*}
    \bm{E} = \sum_{k=1}^{\infty} - \left(-\frac{1}{\mu} \right)^k \bm{E}^{(k)}.
\end{align*}
The layer combination appears naturally and the coefficients decrease quickly. As discussed in \cite{he2020lightgcn}, the layer combination is the key to alleviate the over-smoothing issue and improve performance.

\subsubsection{Neighborhood-based Approaches} The neighborhood-based approaches are often considered as exploiting first-order graph information in the literature discussions \cite{he2020lightgcn}. We consider the following formulation, which utilizes the gram matrix as the similarity matrix \cite{aiolli2013efficient}, i.e., $ \bar{\bm{s}}_u = \bm{r}_u \tilde{\bm{P}}$. Obviously, the corresponding filter is a first-order linear filter
\begin{align*}
    h(\lambda_i) = 1 - \lambda_i.
\end{align*}
and the corresponding spatial GCN is a one-layer GCN. This approach is simple and scalable. However, it lacks higher-order information on the graph. 

\subsubsection{LGCN-IDE} For completeness, we analyze LGCN-IDE \eqref{LGCN-INF}. By eigendecomposition, the corresponding filter takes the form of 
\begin{align*}
    h(\lambda_i) = \sum_{k=0}^{K-1} \beta_k (1 - \lambda_i)^k.
\end{align*}
Since it is still a LightGCN, LGCN-IDE naturally corresponds to multi-layer spatial GCN with a layer combination.

\subsection{A Simple yet Effective Baseline Algorithm}
In this subsection, we develop a simple yet effective baseline algorithm, whose training is as efficient as the inference of LightGCN with a big-O notation. We first analyze the inference computational complexity of LightGCN. We denote the number of non-zero elements in $\bm{R}$ as $\eta$. For a LightGCN with $d$-dimensional embedding, the inference time is $\mathcal{O}(\eta d)$.

The general graph filters require eigendecomposition and thus are not efficient for large-scale recommendation \cite{he2020lightgcn}. Fortunately, there are some graph filters that enjoy a  high computational efficiency, i.e., linear filters and ideal low-pass filters. In order to obtain linear filters, only the normalization is required during training, and thus the training complexity is $\mathcal{O}(\eta)$. A major drawback of the linear filters is that they can hardly obtain a high-order information of the graph. 

For the ideal low-pass filter, only the top-K eigenvectors of $\tilde{\bm{P}}$ are required. Nevertheless, a direct computation for the  top-K eigenvector of $\tilde{\bm{P}}$ is far from computation and memory efficient because $\tilde{\bm{P}}$ is not as sparse as $\bm{R}$. By using the equivalent formulation in \eqref{eq:power_sgmc}, the largest eigenvectors can be computed by \eqref{eq:mf_1}, \eqref{eq:mf_2}, and this iterative algorithm is called the generalized power method (GPM) in the optimization literature \cite{journee2010generalized}. In GPM, we only need to store $\bm{R}$ instead of $\tilde{\bm{P}}$, and the computational complexity is $\mathcal{O}((d\eta + d^3)\log(1/\epsilon))$ where $\epsilon$ is the desired accuracy for the eigenvectors. This algorithm is efficient as long as $d^2 < \eta$. As discussed before, the ideal low-pass filter is equivalent to an infinite layer GCN without layer combination and it suffers from over-smoothing, which means that it lacks a low-order information in the graph. 

As a result, we argue that combining the linear filter and ideal low-pass filter will result in a strong baseline. Specifically, our proposed algorithm, named as Graph
Filter based Collaborative Filtering (GF-CF), has the following form
\begin{align} \label{eq:GF-CF}
    \bm{s}_u = \bm{r}_u\left( \tilde{\bm{R}}^T \tilde{\bm{R}} +  \alpha\bm{D}_I^{-\frac{1}{2}}\bar{\bm{U}}\bar{\bm{U}}^T \bm{D}_I^{\frac{1}{2}} \right),
\end{align}
where $\bm{s}_u$ and $\bm{r}_u$ denote predicted and observed scores, respectively. Likewise, $\bar{\bm{U}}$ is the top-K singular vectors of $\tilde{\bm{R}}$, and $\alpha$ is the tuned parameter. We acknowledge that learning $\alpha$ or transforming \eqref{eq:GF-CF} into GCNs may lead to better performance. Nevertheless, we will demonstrate that \eqref{eq:GF-CF} already achieves the state-of-the-art performance. 

\section{Experiments} \label{sec:exp}
In this section, we first describe the experimental settings, which exactly follow \cite{he2020lightgcn}. Next, we compare our method with the state-of-the-art deep learning methods.
\subsection{Experimental Settings}
To keep the comparison fair, we use the same datasets, the same train/test splitting, and the identical evaluation metric as in \cite{he2020lightgcn}. The statistics of the datasets are listed in Table \ref{tab:dataset}. The evaluation metrics are recall@20 and ndcg@20. 

\begin{table}[t]
\caption{Statistics of the experimented data.}
\label{tab:dataset}
\resizebox{0.48\textwidth}{!}{
\begin{tabular}{l|r|r|r|r}
\hline
\multicolumn{1}{c|}{Dataset} & \multicolumn{1}{c|}{\# User} & \multicolumn{1}{c|}{\# Item} & \multicolumn{1}{c|}{\# Interaction} & \multicolumn{1}{c}{Density} \\ \hline\hline
Gowalla & $29,858$ & $40,981$ & $1,027,370$ & $0.00084$ \\ \hline
Yelp2018 & $31,668$ & $38,048$ & $1,561,406$ & $0.00130$ \\ \hline
Amazon-book & $52,643$ & $91,599$ & $2,984,108$ & $0.00062$ \\ \hline
\end{tabular}}
\end{table}

\subsubsection{Benchmarks}
We follow \cite{he2020lightgcn} to set up the benchmarks.  
\begin{enumerate}
    \item LightGCN \cite{he2020lightgcn}: LightGCN is the state-of-the-art method for CF. Please refer to Section \ref{sec:lightgcn} for a detailed description. 
    \item NGCF \cite{wang2019neural}: NGCF is a nonlinear deep GCN-based method. Besides the components in LightGCN, it contains of feature transformation, and nonlinear activation. 
    \item GRMF and GRMF-norm\cite{rao2015collaborative,he2020lightgcn}: GRMF adds a graph Laplacian regularizer to the training objective of BPR loss in matrix factorization. In GRMF-norm, the normalized Laplacian is adopted instead of the graph Laplacian.
    \item Mult-VAE \cite{liang2018variational}: This is a variational autoencoder based method. The data is assumed to be generated by a multinomial distribution and variational inference is adopted to estimate the parameters. 
\end{enumerate}
In \cite{wang2019neural}, it has been shown that NGCF outperforms GC-MC \cite{berg2017graph}, Pinsage \cite{ying2018graph}, NeuMF \cite{he2017neural}, CMN \cite{ebesu2018collaborative}, MF \cite{rendle2012bpr}, HOP-Rec \cite{yang2018hop} on the same train/test splitting. Thus, we will not include these methods as benchmarks. We also do not compare with full rank models \cite{ning2011slim,steck2019markov} due to the out of memory on Amazon-book dataset. The hyperparameter settings are identical to \cite{he2020lightgcn}.

For the proposed graph filter based methods, we focus on the following two variants: 
\begin{enumerate}
    \item GF-CF: The proposed simple baseline method for CF in \eqref{eq:GF-CF}.
    \item LGCN-IDE: The untrained LightGCN with infinitely dimensional embedding. The closed-form is given in \eqref{LGCN-INF}. 
\end{enumerate}
For the implementation of graph filters, we adopt Scipy \cite{virtanen2020scipy} for sparse operation.

\begin{table}[t!] 
\caption{The comparison of overall performance among GF-CF and competing methods. The performance of benchmarks is reproduced from \cite{he2020lightgcn}.}
\label{tab:overall}
\resizebox{0.48\textwidth}{!}{
\begin{tabular}{l|c c|c c|c c}
\hline
\textbf{Dataset} & \multicolumn{2}{c|}{\textbf{Gowalla}} & \multicolumn{2}{c|}{\textbf{Yelp2018}} & \multicolumn{2}{c}{\textbf{Amazon-book}} \\ \hline
\textbf{Method} & \textbf{recall} & \textbf{ndcg} & \textbf{recall} & \textbf{ndcg} & \textbf{recall} & \textbf{ndcg} \\ \hline\hline
NGCF & 0.1570 & 0.1327 & 0.0579 & 0.0477 & 0.0344 & 0.0263 \\ 
Mult-VAE   & 0.1641	& 0.1335 & 0.0584 & 0.0450 &  0.0407 & 0.0315\\ 
GRMF & 0.1477 & 0.1205 & 0.0571 & 0.0462 & 0.0354 &
0.0270\\ 
GRMF-norm & 0.1557 & 0.1261 & 0.0561 & 0.0454
 & 0.0352 & 0.0269\\
LightGCN & 0.1830 & \textbf{0.1554} & 0.0649 & 0.0530 & 0.0411 & 0.0315 \\\hline
LGCN-IDE & 0.1682 & 0.1347 & 0.0609 & 0.0505 & 0.0612 & 0.0514  \\
GF-CF & \textbf{0.1849} & 0.1518 & \textbf{0.0697} & \textbf{0.0571} & \textbf{0.0710} & \textbf{0.0584} \\\hline
\end{tabular}
}
\end{table}

\begin{table*}[h]
\caption{The comparison of performance and training time of GF-CF and LightGCN.}
\label{tab:time}
\resizebox{0.9\textwidth}{!}{
\begin{tabular}{l|c c c|c c c|c c c}
\hline
\textbf{Dataset} & \multicolumn{3}{c|}{\textbf{Gowalla}} & \multicolumn{3}{c|}{\textbf{Yelp2018}} & \multicolumn{3}{c}{\textbf{Amazon-book}} \\ \hline
\textbf{Method} & \textbf{recall} & \textbf{ndcg} & \textbf{training time} &  \textbf{recall} & \textbf{ndcg} & \textbf{training time}  & \textbf{recall} & \textbf{ndcg} & \textbf{training time} \\ \hline\hline
LightGCN-64 & 0.1830 & 0.1554 & $2.77 \times 10^4$s & 0.0649 & 0.0530 & $5.15 \times 10^{4}$s & 0.0411 & 0.0315 &  $1.27 \times 10^{5}$s \\ 
LightGCN-128   & 0.1878	& 0.1591 & $3.31 \times 10^4$s & 0.0671 &  0.0550 & $5.66 \times 10^{4}$s & 0.0459 & 0.0353 & $1.81 \times 10^{5}$s\\ 
LightGCN-256 & 0.1893 & 0.1606& $4.54 \times 10^{4}$s & 0.0689 & 0.0568 &$8.09 \times 10^{4}$s & 0.0481 & 0.0371 & $2.98 \times 10^5$s \\
LightGCN-512 & 0.1892 & 0.1604 & $7.28 \times 10^{4}$s & 0.0689 & 0.0569 &$1.33 \times 10^{5}$s & 0.0485  & 0.0375 & $5.26 \times 10^5$s \\\hline
GF-CF & 0.1849 & 0.1518 & 30.5s & 0.0697 & 0.0571 & 46.0s & 0.0710 & 0.0584 & 65.8s \\\hline
\end{tabular}
}
\end{table*}

\subsection{Performance Comparison}
The performance of the proposed methods and other benchmarks are shown in Table \ref{tab:overall}. Despite the simplicity, GF-CF achieves competitive or better performance than deep learning-based methods.

\subsubsection{LGCN-IDE versus LightGCN} LGCN-IDE is an untrained LightGCN with an infinitely dimensional embedding. On Gowalla and Yelp2018, which are of small sizes, LightGCN outperforms LGCN-IDE. However, LGCN-IDE outperforms LightGCN by a large margin on the large-scale dataset, i.e., the Amazon-book dataset. In LightGCN, the known scores are compressed into limited dimensional vectors, which restricts the expressiveness. In contrast, in LGCN-IDE, the ratings are directly used as the graph signal without compression. Additionally, LightGCN is trained with a stochastic gradient descent (SGD) while LGCN-IDE has a closed-form solution. As the size of the dataset increases, the optimization by SGD becomes more difficult. We suspect that these two reasons contribute to the large performance gain of LGCN-IDE over LightGCN in the Amazon-book dataset.

\subsubsection{Graph filters versus deep learning-based methods} In Table \ref{tab:overall}, the simple graph filter achieves competitive or better performance compared with deep learning-based methods. LightGCN also outperforms NGCF by removing the non-linear transformations. From the universal approximation theory \cite{hornik1989multilayer}, deep neural networks can approximate linear functions easily. Nevertheless, linear functions are non-trivial to learn for a neural network trained with SGD. A recent theoretical study demonstrates that it is impossible for neural networks with tanh, cosine, or quadratic activation to extrapolate the linear functions well \cite{xu2020neural}. With ReLU activation, A neural network can extrapolate linear functions well if the training data cover all directions (e.g., a hypercube covering the origin) \cite{xu2020neural}, which is not trivial to satisfy in practice. This theoretical result suggests that learning linear functions is a non-trivial task. In addition, deep neural networks do well in extracting complicated features, but CF with implicit feedback is in lack of rich features. Owing to these two factors, the linear models are able to outperform deep models in CF with implicit feedback.

\subsection{Comparison with LightGCN of Large Embedding Dimension} \label{sec:time}
In this subsection, we compare GF-CF with LightGCNs of different embedding dimensions. For the untrained LightGCN, the performance improves significantly with the dimension as shown in Fig. \ref{fig:untrained}. The natural questions are 1) does the performance of trained LightGCN increase significantly as the dimension grows; 2) how does GF-CF perform compared with LightGCN with large embedding dimensions. We validate these questions empirically in Table \ref{tab:time}. The experiments in this subsection are conducted on a server with an Intel Xeon(R) CPU E5-2698 v4 @ 2.20GHz and a Tesla V100 GPU. For the implementation of LightGCN, we download the source code from https://github.com/gusye1234/LightGCN-PyTorch and train $1000$ epochs as the original paper\footnote{We notice that training LightGCN for $400-600$ instead of $1000$ epochs only introduce a slight performance loss, which reduces the training time of LightGCN, but this does not affect our conclusion as we have more than three magnitudes of speedups.}. Due to the excessive training cost, we do not train LightGCN with an embedding dimension of more than $512$. As shown in Table \ref{tab:time}, GF-CF still achieves competitive or higher performance than LightGCN with large embedding dimensions. As the embedding dimension grows, the performance improvement of LightGCN becomes marginal, which is similar to matrix factorization and neural collaborative filtering \cite{rendle2020neural}. The overall training time of GF-CF is even smaller than $1$ training epoch consumed by LightGCN. It demonstrates that GF-CF is a simple but hard-to-beat baseline method for CF.

\section{Related Works}
\label{sec:literature}
\subsection{Collaborative Filtering Methods}
Collaborative filtering (CF) plays a fundamental role in modern recommender systems \cite{covington2016deep}. One popular paradigm is the model-based CF methods. In such methods, the users and items are parameterized by (low-dimensional) vectors and the interactions are reconstructed based on the embeddings and model weights. The classic matrix factorization (MF) maps the ID of users and items as embedding vectors and uses the dot product between embedding vectors as predicted scores. The dot product model can be further improved by using neural networks \cite{he2017neural,tay2018latent}. Another classic model-based CF is to reconstruct the score for an item by a transformation of the scores for other items, from linear auto-encoders (e.g., SLIM \cite{ning2011slim}) to deep auto-encoders (e.g., Multi-VAE \cite{liang2018variational}). Another paradigm is graph-based CF methods. The early works (e.g., Item-rank \cite{gori2007itemrank} and Bi-rank \cite{he2016birank}) exploit the label propagation on graph and belong to the neighborhood-based methods. These methods are often considered as heuristics and inferior to model-based methods due to the lack of training. Recent works address this issue by developing GCN-based methods and train GCNs in an end-to-end manner, e.g., GC-MC \cite{berg2017graph}, NGCF \cite{wang2019neural}, and LightGCN \cite{he2020lightgcn}. 

Notice that the information contained in the sparse rating matrix or graph formulation are identical and GFT is a matrix factorization. In this paper, we unify the two paradigms from the graph signal processing view and identify that the low-pass filters are the underlying key component in the two paradigms. In addition, we show that different paradigms correspond to different low-pass filters and these filters can be incorporated together to improve the performance.

\subsection{Spectral and Spatial GCNs} The spectral GCNs are developed from graph signal processing with learned graph filters, which enjoy theoretical guarantees from graph signal processing theory \cite{ruiz2021graph}. Nevertheless, GFT requires full eigendecomposition, which induces prohibitive computation for large-scale graphs. The spectral CF \cite{zheng2018spectral} and LCF \cite{yu2020graph} belong to this category and thus they cannot be applied on large-scale datasets. To speed up the computation, the spatial GCNs based on 1-hop neighbor propagation were proposed \cite{xu2018powerful}. In each layer of spatial GCNs, only neighborhood aggregations are required, and thus the computational cost is extensively reduced. In the context of recommendation, spatial GCNs contain GCMC \cite{berg2017graph}, NGCF \cite{wang2019neural}, LightGCN \cite{he2020lightgcn}, and PinSage \cite{ying2018graph}. A unique advantage of these methods is the scalability, meaning that they can be applied to large-scale sparse datasets. A recent theoretical study unified the spectral and spatial GCNs and demonstrates that they are all low-pass filters \cite{balcilar2021analyzing}. In the paper, we also unify the classic CF methods via low-pass filtering, which explains the success of GCNs in CF.

\section{Conclusions}
\label{sec:con}
In this paper, we identified the importance of smoothness in the embeddings in a successful recommendation both theoretically and empirically, which bridges CF and graph signal processing theory. Via the lens of graph signal processing, we showed that the neighborhood-based methods, low-rank matrix completion, and linear auto-encoders are all graph convolution with low-pass filters. This further validated the power of graph convolution for recommendation. In addition to our theoretical analysis, we also developed a simple but hard-to-beat baseline algorithm, GF-CF. It was demonstrated that GF-CF achieves competitive or better performance than deep learning-based methods. We believe that the insights of this investigation are inspirational to the principled GCN architecture design for recommender systems. In the future, we will implement the GCNs induced by classic algorithms in Table \ref{tab:unify} and exploit additional information, e.g., social networks and knowledge graphs.

\section{Proofs}

\subsection{Proof of Theorem \ref{prop:untrain}}
\begin{proof}
We first prove that \eqref{eq:rnd_bpr} holds when the mutual coherence \cite{donoho2003optimally}) of the embeddings satisfies
\begin{align} \label{eq:mi}
    \epsilon = M_{\bm{E}^{(0)}} < \sqrt{\frac{N_{\min}}{2N_{\max}^{3} }},
\end{align}
and then show that as $d > \frac{C N_{\max}^{3} \log(|\mathcal{I}| + |\mathcal{U}|)}{N_{\min}}$, \eqref{eq:mi} holds with probability at least $3/4$.
\begin{align*} 
    &\bm{e}^{(1)T}_i \bm{e}^{(0)}_j - \bm{e}_i^{(1)T} \bm{e}^{(0)}_k  \\
    = &\left( \frac{1}{\sqrt{N_i}}  \sum_{l \in \mathcal{N}_i} \frac{1}{\sqrt{N_l}} \bm{e}_l^{(0)} \right)^T (\bm{e}_j^{(0)} - \bm{e}_k^{(0)}) \\
    \geq &\frac{1}{\sqrt{N_i}}  \left(\frac{1}{\sqrt{N_j} } - \frac{N_i- 1}{\sqrt{N_{\min}}}\epsilon - \frac{N_i}{\sqrt{N_{\min}}}\epsilon  \right) \\
    \geq &\frac{1}{\sqrt{N_i}}\left(\frac{1}{\sqrt{N_{\max} } } - \frac{2N_{\max}}{\sqrt{N_{\min}}}\epsilon\right) \overset{(a)}{>} 0
\end{align*}
where (a) follows the assumption that $\epsilon < \sqrt{\frac{N_{\min}}{2N_{\max}^{3} }}$. With Lemma \ref{lem:mi}, we see that $\epsilon < \sqrt{\frac{N_{\min}}{2N_{\max}^{3} }}$ as $d > \frac{C N_{\max}^{3} \log(|\mathcal{I}| + |\mathcal{U}|)}{N_{\min}}$.
\end{proof}

\begin{lemma} \label{lem:mi} (Theorem 3.5 in \cite{Wright-Ma-2021})
    Let $\bm{A} \in \mathbb{R}^{n \times m}$ with rows i.i.d. chosen from the uniform distribution on the sphere. Then with probability at least $3/4$,
    \begin{align*}
        M_{\bm{A}} \leq C \sqrt{\frac{\log n}{m}}.
    \end{align*}
    where $C$ is an absolute constant.
\end{lemma}

\subsection{Proof of Theorem \ref{thm:inf_lgcn}}
\begin{proof}
We first separate the embeddings $\bm{E}^{(k)}$ into user embeddings $\bm{U}^{(k)}$ and item embeddings $\bm{V}^{(k)}$, and the individual update is given by
\begin{align*}
    \bm{U}^{(k+1)} = \tilde{\bm{R}} \bm{V}^{(k)}, \quad \bm{V}^{(k+1)} = \tilde{\bm{R}}^T \bm{U}^{(k)}.
\end{align*}

The final embeddings are 
\begin{align*}
    V &=  \left(\alpha_0\bm{V}^{(0)} + \alpha_1 \tilde{\bm{R}}^T\bm{U}^{(0)} + \alpha_2\tilde{\bm{R}}^T\tilde{\bm{R}}\bm{V}^{(0)} + \alpha_3 \tilde{\bm{R}}^T\tilde{\bm{R}}\tilde{\bm{R}}^T\bm{U}^{(0)} + \cdots\right) \\
    &= 
    \left(\sum_{i=0}^{2i \leq K} \alpha_{2i}(\tilde{\bm{R}}^T\tilde{\bm{R}})^{i} \bm{V}^{(0)} + \sum_{i = 0}^{2i+1 \leq K} \alpha_{2i+1}(\tilde{\bm{R}}^T \tilde{\bm{R}})^{i} \tilde{\bm{R}}^T \bm{U}^{(0)} \right)
\end{align*}
and $\bm{U}$ can be computed similarly.

The final prediction of untrained LightGCN with infinitely dimensional embedding is given by
\begin{align*} \label{eq:rec}
    \bm{S} = \bm{U} \bm{V}^T = &\left(\sum_{i=0}^{2i \leq K}\alpha_{2i} (\tilde{\bm{R}}^T\tilde{\bm{R}})^{i} \bm{U}^{(0)} + \sum_{i = 0}^{2i+1 \leq K} \alpha_{2i+1}(\tilde{\bm{R}} \tilde{\bm{R}}^T)^{i} \tilde{\bm{R}} \bm{V}^{(0)} \right)  \\
    &\cdot \left(\sum_{i=0}^{2i \leq K} \alpha_{2i}(\tilde{\bm{R}}^T\tilde{\bm{R}})^{i} \bm{V}^{(0)} + \sum_{i = 0}^{2i+1 \leq K} \alpha_{2i+1}(\tilde{\bm{R}}^T \tilde{\bm{R}})^{i} \tilde{\bm{R}}^T \bm{U}^{(0)} \right)^T \\
\end{align*}

For a pair of matrices $\bm{X}, \bm{Y}$, if the rows of $\bm{X} \in \mathbb{R}^{* \times d}, \bm{Y}^{(0)} \in \mathbb{R}^{* \times d} $ follow independently identical distribution, due to the linearity of dot product, we have $\lim_{d \rightarrow \infty} \bm{X} \bm{Y}^T = \mathbb{E} [\bm{x}_1 \bm{y}_1^T]$, where $\bm{x}_1$ (resp. $\bm{y}_1$) denotes the first column of $\bm{X}$ (resp. $\bm{Y}$).

Thus, as $d \rightarrow \infty$, we have 
\begin{align*}
    \lim_{d \rightarrow \infty} \bm{S} = \mathbb{E}_{\bm{U}^{(0)}, \bm{V}^{(0)}} [\bm{S}] = \sum_{k=0}^{K-1} \beta_k \tilde{\bm{R}} (\tilde{\bm{R}}^T\tilde{\bm{R}})^k.
\end{align*}
where $\beta_k$ depends on $[\alpha_k]_{i = 0,\cdots,K}$.

For a given user $u$, the estimated scores is shown as \\ $\bm{s}_u = \tilde{\bm{r}}_u \sum_{k=0}^{K-1} \beta_k  (\tilde{\bm{R}}^T\tilde{\bm{R}})^k$.

\end{proof}

\subsection{Proof of Theorem \ref{thm:P}}
\begin{proof}
We observe that $\tilde{\bm{A}}^2 = \begin{bmatrix}
    \tilde{\bm{R}}\tilde{\bm{R}}^T & \bm{0} \\
    \bm{0} & \tilde{\bm{R}}^T\tilde{\bm{R}}
    \end{bmatrix}$.
As $\tilde{\bm{A}}^2$ is block diagonal, eigenvalues of $\tilde{\bm{A}}^2$ are a concatenation of eigenvalues of $\tilde{\bm{R}}\tilde{\bm{R}}^T$ and $\tilde{\bm{R}}^T\tilde{\bm{R}}$. For the largest eigenvalue, we have 
\begin{align*}
    \lambda_{\max} (\tilde{\bm{R}}^T\tilde{\bm{R}}) = \lambda_{\max} (\tilde{\bm{R}}\tilde{\bm{R}}^T) = \lambda_{\max}(\tilde{\bm{A}}^2) = \lambda_{\max}(\tilde{\bm{A}})^2 \overset{(a)}{=} 1.
\end{align*}
where (a) follows Lemma \ref{lem:eigs}. As $\tilde{\bm{R}}^T\tilde{\bm{R}}$ is positive semi-definite, $\lambda_{\min}(\tilde{\bm{R}}^T\tilde{\bm{R}}) \geq 0$. This finishes the proof.
\end{proof}

\begin{lemma} \label{lem:eigs}
    Let $\lambda_1 \leq \lambda_2 \leq \cdots \lambda_{|\mathcal{I}| + |\mathcal{U}|}$ be eigenvalues of $\tilde{\bm{A}}$. Then $-1 \leq \lambda_1 \leq \lambda_2 \leq \cdots \lambda_{|\mathcal{I}| + |\mathcal{U}|} = 1$.
\end{lemma}
\begin{proof}
First, observing that $\forall \bm{x} \in \mathbb{S}^{n-1}$, we have
\begin{align*}
    \bm{x}^T(\bm{I} - \tilde{\bm{A}})\bm{x} = \sum_{i,j} A_{i,j} \left(\frac{x(i)}{\sqrt{d(i)}} - \frac{x(j)}{\sqrt{d(j)}} \right)^2 \geq 0.
\end{align*}
Thus, $-1 \leq \bm{x}^T(\bm{I} - \tilde{\bm{A}})\bm{x} \leq 1$. Furthermore, using the vector $\bm{x} = \bm{D}^{\frac{1}{2}}\bm{1}$, we get 
\begin{align*}
    \tilde{\bm{A}} \bm{x} = \bm{D}^{-\frac{1}{2}}\bm{A}\bm{D}^{-\frac{1}{2}} \bm{D}^{\frac{1}{2}} \bm{1} = \bm{D}^{-\frac{1}{2}}\bm{A} \bm{1} = \bm{D}^{-\frac{1}{2}} \text{diag}(\bm{D}) = \bm{D}^{\frac{1}{2}}\bm{1}.
\end{align*}
This implies that the largest eigenvalue of $\tilde{\bm{A}}$ is $1$.
\end{proof}

\bibliographystyle{ACM-Reference-Format}
\bibliography{sample-base}

\end{document}